\documentclass[conference]{IEEEtran}

\usepackage{graphicx}
\usepackage{amsfonts}
\usepackage{amsmath}
\usepackage{amssymb}
\usepackage{color}

\usepackage{bm} 

\newtheorem{thm}{Theorem}

\newtheorem{lem}[thm]{Lemma}
\newtheorem{proof}[thm]{proof}

\newtheorem{defn}[thm]{Definition}

\newtheorem{exam}[thm]{Example}

\ifCLASSINFOpdf

\else

\fi

\begin{document}

\title{Binary Locally Repairable Codes ---Sequential Repair for
Multiple Erasures}
\author{\IEEEauthorblockN{Wentu Song}
\IEEEauthorblockA{Singapore University of Technology\\ and Design,
Singapore\\
Email: wentu\_song@sutd.edu.sg} \and \IEEEauthorblockN{Chau Yuen}
\IEEEauthorblockA{Singapore University of Technology\\ and
Design, Singapore\\
Email: yuenchau@sutd.edu.sg}} \maketitle

\begin{abstract}
Locally repairable codes (LRC) for distribute storage allow two
approaches to locally repair multiple failed nodes: 1) parallel
approach, by which each newcomer access a set of $r$ live nodes
$(r$ is the repair locality$)$ to download data and recover the
lost packet; and 2) sequential approach, by which the newcomers
are properly ordered and each newcomer access a set of $r$ other
nodes, which can be either a live node or a newcomer ordered
before it. An $[n,k]$ linear code with locality $r$ and allows
local repair for up to $t$ failed nodes by sequential approach is
called an $(n,k,r,t)$-exact locally repairable code (ELRC).

In this paper, we present a family of binary codes which is
equivalent to the direct product of $m$ copies of the $[r+1,r]$
single-parity-check code. We prove that such codes are
$(n,k,r,t)$-ELRC with $n=(r+1)^m,k=r^m$ and $t=2^m-1$, which
implies that they permit local repair for up to $2^m-1$ erasures
by sequential approach. Our result shows that the sequential
approach has much bigger advantage than parallel approach.
\end{abstract}


\IEEEpeerreviewmaketitle

\section{Introduction}
In a distributed storage system (DSS), data is stored through a
large, distributed network of storage nodes. To maintain the data
reliability in the presence of node failures, the system should
have the ability of \emph{node repair}. That is, when some of the
storage nodes fail, each failed node is replaced by a
\emph{newcomer} where the lost packet is recovered and stored
again.

Various coding techniques are employed by modern DSS to improve
system performance, among which locally repairable codes (LRC) aim
to minimize the repair locality, i.e. the number of disk accesses
required for single node repair
\cite{Papail122}$-$\cite{Papail121}.

The $i$th coordinate of an $[n, k]$ linear code $\mathcal C~($also
called the $i$th code symbol of $\mathcal C)$ is said to have
\emph{locality} $r$, if its value is computable from the values of
a set of at most $r$ other coordinates of $\mathcal C~($called a
repair set of $i)$. Codes with all code symbols having locality
$r~(r<k)$ are called locally repairable codes. In a DSS with an
LRC $\mathcal C$ as the storage code, the data packet stored in
each storage node is a code symbol of $\mathcal C$ and any single
failed node can be ``locally and exactly repaired" in the sense
that the newcomer can recover the lost data by accessing at most
$r$ other nodes, where $r$ is the locality of $\mathcal C$.

To handle the problem of local repair for multiple failed nodes,
some special subclasses of LRCs are investigated, such as: a)
Codes with all-symbol locality $(r,t+1)$, also called $(r,t+1)_a$
codes, in which each code symbol is contained in a local code of
length at most $r+t$ and minimum distance at least $t+1$
\cite{Prakash12}; b) Codes with all-symbol locality $r$ and
availability $t$, in which each code symbol has $t$ pairwise
disjoint repair sets with locality $r$ \cite{Wang14,Rawat14}; c)
Codes with $(r,t)$-locality, in which each subset of $t$ code
symbols can be cooperatively repaired from at most $r$ other code
symbols \cite{Rawat-14} $($For convenience, in the following, we
will call such codes as $(r,t)$-CLRC.$)$; d) Codes with overall
local repair tolerance $t$, in which for any $E\subseteq[n]$ of
size $t$ and any $i\in E$, the $i$th code symbol has a repair set
contained in $[n]\backslash E$ and with locality $r$
\cite{Pamies13}. Clearly, these four subclasses of LRC permit
local repair for up to $t$ failed nodes by \emph{parallel
approach} --- each newcomer can access $r$ live nodes to recover
the corresponding lost packet. We also call $t$ as the erasure
tolerance of such codes.

For $(r,\delta)_a$ codes and $(r,t)$-CLRC, the code rate satisfies
(e.g., see \cite{Wentu14} and \cite{Rawat-14}):
\begin{align}\label{rate-bd-1}\frac{k}{n}\leq
\frac{r}{r+t}.\end{align} For codes with locality $r$ and
availability $t$, it was proved in \cite{Tamo14} that the code
rate satisfies:
\begin{align}\label{rate-bd-2}\frac{k}{n}\leq
\frac{1}{\prod_{j=1}^{t}(1+\frac{1}{jr})}.\end{align} However, for
$t\geq 2$, it is not known whether the code rate bound
\eqref{rate-bd-2} is achievable. Recent work by Wang et al.
\cite{Wang15} shows that for any positive integers $r$ and $t$,
there exist codes with locality $r$ and availability $t$ over the
binary field with code rate $\frac{r}{r+t}$. Unfortunately, such
codes do not achieve the bound \eqref{rate-bd-2} for $t\geq 2$.
The problem of constructing codes with locality $r$ and
availability $t\geq 2$ that achieve the optimal code rate is still
an open problem.

A more general way to locally repair $t~(t\geq 2)$ failed nodes is
the \emph{sequential approach}, by which the $t$ newcomers can be
properly ordered in a sequence and, to recover the lost packet,
each newcomer can access $r$ other nodes, each of which can be a
live node or a newcomer ordered before it \cite{Prakash-14,
Wentu-15}. In \cite{Wentu-15}, an $[n,k]$ linear code that has
locality $r$ and permit local repair for up to $t$ failed nodes by
sequential approach is called an $(n,k,r,t)$-\emph{exact locally
repairable code} (ELRC). Clearly, the four subclasses of LRC,
i.e., $(r,\delta)_a$ codes, $(r,t)$-CLRC, codes with locality $r$
and availability $t$, and codes with overall local repair
tolerance $t$, are all $(n,k,r,t)$-ELRC. Potentially, the
sequential approach allows us to design codes with improved
parameter properties than the parallel approach.

\renewcommand\figurename{Fig}
\begin{figure}[htbp]
\begin{center}
\includegraphics[height=2.3cm]{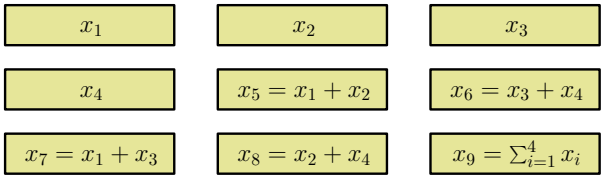}
\end{center}
\vspace{-0.2cm}\caption{A $(9,4,2,3)$-ELRC. }\label{exam-code-1}
\end{figure}

\begin{exam}\label{exam-ELRC-1}
As an example of sequential approach, consider the code
illustrated in Fig. \ref{exam-code-1}, where $x_1,\cdots,x_4$ are
information symbols and $x_5,\cdots,x_9$ are parity symbols. We
can check that $x_1=x_2+x_5=x_3+x_7$. So $\{x_2,x_5\}$ and
$\{x_3,x_7\}$ are two disjoint repair sets of $x_1$. Similarly, we
can find two disjoint repair sets for each of $x_2,\cdots, x_9$.
The repair set of each code symbol is illustrated in Fig.
\ref{exam-code-1-1}. So this code has locality $2$ and
availability $2$. Hence, it permits local repair for up to $2$
erasures by the parallel approach.
\end{exam}

\renewcommand\figurename{Fig}
\begin{figure}[htbp]
\begin{center}
\includegraphics[height=3.4cm]{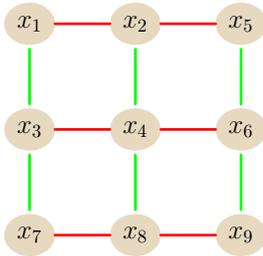}
\end{center}
\vspace{-0.2cm}\caption{Repair relation of code symbols of the
code in Fig. \ref{exam-code-1}: Each line (red line or green line)
contains $3$ symbols and any symbol on a line can be computed from
the other symbols on the same line. Note that each symbol belongs
to two lines --- a red line and a green line, hence has two repair
sets.}\label{exam-code-1-1}
\end{figure}

However, we can check that this code is a $(9,4,2,3)$-ELRC --- it
permits local repair for up to $3$ failed nodes by sequential
approach. For example, if $x_1,x_5,x_7$ are lost, then we can
repair them by the following sequence of equations: $x_5=x_6+x_9$,
$x_7=x_8+x_9$ and $x_1=x_2+x_5$. During the repair process, $x_5$
is repaired before $x_1$. Once $x_5$ is repaired, it can be used
to repair $x_1$. Hence, the repair process is feasible. Note that
$x_1,x_5,x_7$ can't be repaired by parallel approach because both
the two repair sets of $x_1$ contain a lost symbol.

Most existing works about LRC focus on parallel repair approach
\cite{Pamies13}$-$\cite{Wentu14}. In the field of $(n,k,r,t)$-ELRC
$($i.e., LRC with sequential repair approach$)$, only for
$t\in\{2,3\}$ is investigated \cite{Prakash-14, Wentu-15}.

For $(n,k,r,t=2)$-ELRC, the code rate satisfies \cite{Prakash-14}:
\begin{align}\label{rate-bd-3}\frac{k}{n}\leq
\frac{r}{r+2}.\end{align} An upper bound for the minimum distance
of such codes and a construction of codes achieving the minimum
distance bound were also given in \cite{Prakash-14}.

For $(n,k,r,t=3)$-ELRC, it was proved in \cite{Wentu-15} that the
code length $n$ satisfies:
\begin{align}\label{rate-bd-4}n\geq
k+\left\lceil\frac{2k+\lceil\frac{k}{r}\rceil}{r}\right\rceil\end{align}
and there exist codes with code length meet this bound. However,
for $t\geq 3$, no result is known about the minimum distance
bound, and for $t\geq 4$, no result is known about the code rate
bound. Construction of $(n,k,r,t\geq 4)$-ELRC is not seen in
literature either.

In this paper, we prove that for any given positive integers
$r~(r\geq 2)$ and $m$, the direct product of $m$ copies of the
$[r+1,r]$ single-parity-check code is an $(n,k,r,t=2^m-1)$-ELRC.
So such code permits local repair for up to $t=2^m-1$ erasures by
sequential approach. The code rate of such codes is shown to be
much larger than $(r,\delta)_a$ codes and $(r,t)$-CLRC. Moreover,
it was pointed out in \cite{Tamo14} that such code has locality
$r$ and availability $m$, which implies that it permits local
repair for up to only $m$ failed nodes by parallel approach.
Hence, our result shows that sequential approach has much bigger
advantage than parallel approach for such codes.

The rest of this paper is organized as follows. In Section
\uppercase\expandafter{\romannumeral 2}, we state the formal
definition of $(n,k,r,t)$-ELRC. In section
\uppercase\expandafter{\romannumeral 3}, we give a method to
construct codes that are equivalent to the direct product codes
and present our main theorem. We prove the main theorem in Section
\uppercase\expandafter{\romannumeral 4} and conclude the paper in
Section \uppercase\expandafter{\romannumeral 5}.

\section{Preliminary}
For any set $A$, we use $|A|$ to denote the size $($i.e., the
number of elements$)$ of $A$. A set $B$ is called an $r$-subset of
$A$ if $B\subseteq A$ and $|B|=r$. For any positive integer $n$,
we denote $$[n]:=\{1,2,\cdots,n\}.$$

An $[n,k]$ linear code over the finite field $\mathbb F$ is a
$k$-dimensional subspace of the vector space $\mathbb F^n$, where
$n,k$ are positive integers and $k\leq n$.

In this section, we present the formal definition of
$(n,k,r,t)$-exact locally repairable code (ELRC). More details can
be found in \cite{Wentu-15}.

Let $\mathcal C$ be an $[n,k]$ linear code over the field $\mathbb
F$. If there is no confusion in the context, we will omit the base
field $\mathbb F$ and only say that $\mathcal C$ is an $[n,k]$
linear code. A $k$-subset $S$ of $[n]$ is called an
\emph{information set} of $\mathcal C$ if for all codeword
$x=(x_1,x_2,\cdots,x_n)\in\mathcal C$ and all $i\in[n]$,
$x_i=\sum_{j\in S}a_{i,j} x_j$, where all $a_{i,j}\in\mathbb F$
and are independent of $x$. The code symbols in $\{x_j, j\in S\}$
are called \emph{information symbol} of $\mathcal C$. In contrast,
code symbols in $\{x_i, i\in [n]\backslash S\}$ are called
\emph{parity symbol} of $\mathcal C$. An $[n,k]$ linear code has
at least one information set.

\begin{defn}\label{def-lrs}
Let $i\in[n]$ and $R\subseteq[n]\backslash\{i\}$. The subset $R$
is called an $(r,\mathcal C)$-\emph{repair set} of $i$ if $|R|\leq
r$ and $x_i=\sum_{j\in R}a_jx_j$ for all
$x=(x_1,x_2,\cdots,x_n)\in\mathcal C$, where all $a_j\in\mathbb F$
and are independent of $x$.
\end{defn}

\begin{defn}\label{r-compute} Let $E$ be a $t$-subset of $[n]$ and
$\overline{E}=[n]\backslash E$. The code $\mathcal C$ is said to
be $(E,r)$-repairable if there exists an index of $E$, say
$E=\{i_1,\cdots,i_t\}$, and a collection of subsets
$$\{R_{\ell}\subseteq\overline{E}\cup\{i_1,\cdots,i_{\ell-1}\};
|R_{\ell}|\leq r, \ell\in[t]\}$$ such that for each $\ell\in[t]$,
$R_{\ell}$ is an $(r,\mathcal C)$-repair set of $i_\ell$.
\end{defn}

\begin{defn}\label{e-lrc}
An $(n,k,r,t)$-\emph{exact locally repairable code (ELRC)} is an
$[n,k]$ linear code $\mathcal C$ such that for each
$E\subseteq[n]$ of size $|E|\leq t$, $\mathcal C$ is
$(E,r)$-repairable.
\end{defn}

By Definition \ref{r-compute} and \ref {e-lrc}, if a DSS uses an
$(n,k,r,t)$-ELRC as the storage code, then any $t'~(t'\leq t)$
failed nodes can be locally repaired by sequential approach.

The following lemma gives a seemingly simpler characterization for
$(n,k,r,t)$-ELRC.

\begin{lem}[\cite{Wentu-15}, Lemma 6]\label{lem-ELRC}
An $[n,k]$ linear code $\mathcal C$ is an $(n,k,r,t)$-ELRC if and
only if for any $E\subseteq[n]$ of size $0<|E|\leq t$, there
exists an $i\in E$ such that $i$ has an $(r,\mathcal C)$-repair
set contained in $[n]\backslash E$.
\end{lem}

In the following, if $R$ is an $(r,\mathcal C)$-repair set of $i$,
we will omit the prefix $(r,\mathcal C)$ and only say that $R$ is
a repair set of $i$.

\section{Code Construction}
Let $r,m$ be two positive integers such that $r\geq 2$. Let
$n=(r+1)^m$ and $k=r^m$. We will construct a binary $[n,k]$ linear
code that is equivalent to the direct product of $m$ copies of the
$[r+1,r]$ single-parity-check code. Moreover, we will show that
such code is an $(n,k,r,t)$-ELRC, where $t=2^m-1$.

In the following, we will denote $$\mathbb
Z_r=\{0,1,\cdots,r-1\}$$ and $$\mathbb
Z_r^m=\{(\lambda_1,\cdots,\lambda_m);
\lambda_1,\cdots,\lambda_m\in\mathbb Z_r\}.$$ That is, $\mathbb
Z_r^m$ is the Cartesian product of $m$ copies of $\mathbb Z_r$.
Similarly, we denote $$\mathbb Z_{r+1}=\{0,1,\cdots,r\}$$ and
$$\mathbb Z_{r+1}^m=\{(\lambda_1,\cdots,\lambda_m);
\lambda_1,\cdots,\lambda_m\in\mathbb Z_{r+1}\}.$$

Then $\mathbb Z_{r}\subseteq\mathbb Z_{r+1}$ and $\mathbb
Z_{r}^m\subseteq\mathbb Z_{r+1}^m$. To describe the code
construction method, we need the following two notations:

For each $\alpha=(\lambda_1,\cdots,\lambda_m)\in\mathbb
Z_{r+1}^m\backslash\mathbb Z_r^m$, denote
\begin{align}\label{T-alf}
T(\alpha)=\{j\in[m]; \lambda_j\in\mathbb Z_r\}\end{align} and
\begin{align}\label{L-alf}\mathcal L(\alpha)
=\{(\mu_1,\cdots,\mu_m)\in\mathbb Z_{r}^m; \mu_j=\lambda_j,
\forall j\in T(\alpha)\}.\end{align}

For example, let $r=2$, $m=6$. For $\alpha=(0,1,2,0,2,2)\in\mathbb
Z_{3}^6$, we have $T(\alpha)=\{1,2,4\}$ and $$\mathcal
L(\alpha)=\{(0,1,\lambda_3,0,\lambda_5,\lambda_6); \lambda_3,
\lambda_5,\lambda_6\in\mathbb Z_2\}.$$

Clearly, for each $\alpha=(\lambda_1,\cdots,\lambda_m)\in\mathbb
Z_{r+1}^m\backslash\mathbb Z_r^m$, $T(\alpha)$ is a proper subset
of $[m]$ and $\mathcal L(\alpha)$ is a non-empty subset of
$\mathbb Z_r^m$. Moreover, if $\alpha=(r,\cdots,r)$, then
$T(\alpha)=T(r,\cdots,r)=\emptyset$ and $\mathcal
L(\alpha)=\mathcal L(r,\cdots,r)=\mathbb Z_{r}^m.$

\vspace{0.1cm}Let $n=(r+1)^m$ and $k=r^m$. Let
$H=(h_{\alpha,\beta})$ be an $(n-k)\times n$ binary matrix whose
rows are indexed by $\mathbb Z_{r+1}^m\backslash\mathbb Z_r^m$ and
columns are indexed by $\mathbb Z_{r+1}^m$ such that
\begin{equation}\label{def-H}
h_{\alpha,\beta}=\left\{\begin{aligned}
&1, ~ ~\text{if}~\beta\in\mathcal L(\alpha)\cup\{\alpha\};\\
&0, ~ ~\text{Otherwise}.\\
\end{aligned} \right.
\end{equation}
Clearly, the submatrix $H_1$ formed by the columns of $H$ that are
indexed by $\mathbb Z_{r+1}^m\backslash\mathbb Z_r^m$ is a
permutation matrix of order $n-k$. So $\text{rank}(H)=n-k$.

Let $\mathcal C$ be the binary code with a parity check matrix
$H$. Then $\mathcal C$ is an $[n,k]$ linear code. Clearly,
$\mathcal C$ is just the $[r+1,r]$ single-parity-check code for
$m=1$ and the square code constructed in \cite{Wang14} for $m=2$.
In general, it is not difficult to prove that the code $\mathcal
C$ is equivalent to the direct product of $m$ copies of the
$[r+1,r]$ single-parity-check code. Moreover, we have the
following theorem.

\vspace{0.1cm}\begin{thm}\label{main-th} The code $\mathcal C$
which has a parity check matrix $H$ is an $(n,k,r,t)$-ELRC, where
$t=2^m-1$.
\end{thm}

\vspace{0.1cm} It is easy to see that the code rate of
$(n,k,r,t)$-ELRC obtained by the above construction is much larger
than the bound \eqref{rate-bd-1}. Comparison of code length of
$(n,k,r,t)$-ELRC with $(r,t+1)_a$ codes and $(r,t)$-CLRC for $r=2$
and $m\in\{2,3,4,5\}$ is given in Table 1, from which we can see
that the code rate of $(n,k,r,t)$-ELRC is much larger than
$(r,t+1)_a$ codes and $(r,t)$-CLRC for the same $r$ and $t$.

Moreover, it was pointed out in \cite{Tamo14} that the direct
product of $m$ copies of the $[r+1,r]$ single-parity-check code
has locality $r$ and availability $m$, which implies that
$\mathcal C$ permits local repair for up to $m$ erasures by
parallel approach. Note that Theorem \ref{main-th} shows that
$\mathcal C$ permits locally repair for up to $2^m-1$ erasures by
the sequential approach, which is much larger than $m$ for $m\geq
2$. Hence, our result shows that sequential approach has much
bigger advantage than parallel approach for LRC. Table 2 is the
comparison of erasure tolerance of the constructed code for
sequential approach and parallel approach, where we assume $r=2$.

\vspace{0.2cm}\begin{center}
\begin{tabular}{|p{0.45cm}|p{0.45cm}|p{0.45cm}|p{1.64cm}|p{1.64cm}|p{1.64cm}|}
\hline \small{$m$} & \small{$t$} & \small{$k$}  & \small{Code
length of $(r,t)$-ELRC} & \small{Code length of $(r,t+1)_a$ codes}
& \small{Code length of $(r,t)$-CLRC}\\
\hline \small{$2$} & \small{$3$}  & \small{$4$}  & \small{$9$}   & \small{$\geq 10$}  & \small{$\geq 10$} \\
\hline \small{$3$} & \small{$7$}  & \small{$8$}  & \small{$27$}  & \small{$\geq 36$}  & \small{$\geq 36$} \\
\hline \small{$4$} & \small{$15$} & \small{$16$} & \small{$81$}  & \small{$\geq 136$} & \small{$\geq 136$} \\
\hline \small{$5$} & \small{$31$} & \small{$32$} & \small{$243$} & \small{$\geq 528$} & \small{$\geq 528$} \\
\hline
\end{tabular}\\
\vspace{0.15cm}\footnotesize{Table 1. Comparison of code length of
three subclasses of LRCs for $r=2$.}
\end{center}

\vspace{0.2cm}\begin{center}
\begin{tabular}{|p{0.5cm}|p{0.5cm}|p{0.5cm}|p{2.2cm}|p{2.2cm}|}
\hline \small{$m$} & \small{$k$}  & \small{$n$}   & \small{Erasure
tolerance by
sequential repair approach} & \small{Erasure tolerance by parallel repair approach}\\
\hline \small{$2$} & \small{$4$}  & \small{$9$}   & \small{$3$}  & \small{$2$} \\
\hline \small{$3$} & \small{$8$}  & \small{$27$}  & \small{$7$}  & \small{$3$} \\
\hline \small{$4$} & \small{$16$} & \small{$81$}  & \small{$15$} & \small{$4$} \\
\hline \small{$5$} & \small{$32$} & \small{$243$} & \small{$31$} & \small{$5$} \\
\hline
\end{tabular}\\
\vspace{0.15cm}\footnotesize{Table 2. Comparison of erasure
tolerance of the constructed code with $r=2$: sequential approach
and parallel approach.}
\end{center}

\vspace{0.1cm}The proof of Theorem \ref{main-th} will be given in
the next section. We now give an example of the above
construction.

\renewcommand\figurename{Fig}
\begin{figure}[htbp]
\begin{center}
\includegraphics[height=7.0cm]{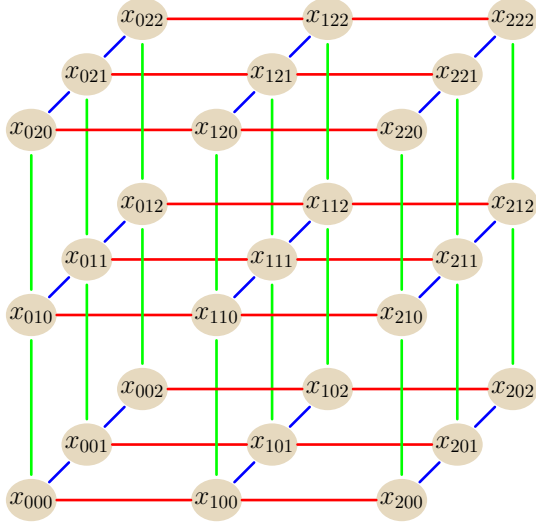}
\end{center}
\caption{Repair relation of code symbols of the code in Example
\ref{eg-code}. We use $\mathbb Z_3^3$ to index the coordinates
and, to simplify notation, use $x_{\lambda_1,\lambda_2,\lambda_3}$
to denote the code symbol $x_{(\lambda_1,\lambda_2,\lambda_3)}$
for each $(\lambda_1,\lambda_2,\lambda_3)\in\mathbb Z_3^3$.}
\label{exam-code-2}
\end{figure}

\begin{exam}\label{eg-code}
Let $r=2$ and $m=3$. Then $k=8$ and $n=27$. We can construct a
matrix $H$ and a binary $[27,8]$ linear code $\mathcal C$ by the
above method. Similar to Fig. \ref{exam-code-1-1}, we can
illustrate the repair set of each code symbol of $\mathcal C$ by
Fig. \ref{exam-code-2}. More details can be seen in Lemma
\ref{rp-set}. We will show that $\mathcal C$ is an
$(n,k,r,t=7)$-ELRC. That is, any $t'\leq 7$ code symbols of
$\mathcal C$ can be locally repaired by other code symbols of
$\mathcal C$. For example, suppose $E=\{(020),(120),(010),(110)$,
$(021),(121),(011)\}$. Then the code symbols in $E$ can be locally
repaired by the following sequence of equation:
$x_{011}=x_{111}+x_{211}$, $x_{121}=x_{111}+x_{101}$,
$x_{021}=x_{121}+x_{221}$, $x_{020}=x_{021}+x_{022}$,
$x_{120}=x_{121}+x_{122}$, $x_{010}=x_{011}+x_{012}$ and
$x_{110}=x_{111}+x_{112}$. In general, this claim can be checked
as follows.
\end{exam}

We partition the index set $\mathbb Z_{r+1}^3=\mathbb Z_{3}^3$
into three subsets
$$I_j=\{(\lambda_1,\lambda_2,\lambda_3);
\lambda_1,\lambda_2\in\mathbb Z_{r+1} \text{~and~}\lambda_3=j\},
j=0,1,2.$$ For example, $I_0=\{(000), (010), (020), (100), (110),
(120)$, $(200), (210), (220)\}.$ For each $j\in\{0,1,2\}$, from
Fig. \ref{exam-code-2}, the repair relation of code symbols in
$I_j$ is the same as code symbols in Fig. \ref{exam-code-1-1}. So
any $t'\leq 3$ code symbols in $I_j$ can be locally repaired by
other code symbols in $I_j$. Now, suppose $E\subseteq\mathbb
Z_{3}^3$ of size $|E|\leq 7$. Then there exist at most one
$j\in\{0,1,2\}$ such that $|E\cap I_j|> 3$. For those $j$ such
that $|E\cap I_j|\leq 3$, code symbols in $E\cap I_j$ can be
locally repaired by other code symbols in $E\cap I_j$. Finally, if
there exist a $j_0$ such that $|E\cap I_{j_0}|> 3$. Then each code
symbol in $E\cap I_{j_0}$ can be locally repaired by code symbols
in $I_{j_1}\cup I_{j_2}$, where
$\{j_1,j_2\}=\{1,2,3\}\backslash\{j_0\}$. Hence, all code symbols
in $E$ can be locally repaired by sequential approach.

\section{Proof of Theorem \ref{main-th}}
In this section, we prove Theorem \ref{main-th}. The basic idea of
the proof is the same as in Example \ref{eg-code}.

Before proving Theorem \ref{main-th}, we first need to prove a
lemma. For each $\alpha=(\lambda_1,\cdots,\lambda_m)\in\mathbb
Z_{r+1}^m$ and each $i\in[m]$, denote
\begin{align}\label{mL_alf}
L^{(i)}_\alpha=\{(\mu_1,\cdots,\mu_m)\in\mathbb Z_{r+1}^m;
\mu_j=\lambda_j, \forall j\in[m]\backslash\{i\}\}.\end{align} Then
we have the following lemma.
\begin{lem}\label{rp-set}
For each $\alpha=(\lambda_1,\cdots,\lambda_m)\in\mathbb Z_{r+1}^m$
and $i\in[m]$, the subset $L^{(i)}_\alpha\backslash\{\alpha\}$ is
a repair set of $\alpha$.
\end{lem}
\begin{proof}
To simplify notation, we assume $i=1$. Then by assumption of this
lemma, we have
\begin{align*}
L^{(1)}_\alpha=\{(\lambda_1',\lambda_2,\cdots,\lambda_m)\in\mathbb
Z_{r+1}^m; \lambda_1'\in\mathbb Z_{r+1}\}.\end{align*} For each
$\lambda_1'\in\mathbb Z_{r+1}$, denote
$\alpha_{\lambda_1'}=(\lambda_1',\lambda_2,\cdots,\lambda_m)$.
Then $L^{(1)}_\alpha=\{\alpha_{0},\alpha_{1}, \cdots,\alpha_{r}\}$
and $\alpha=\alpha_{\lambda_1}\in L^{(1)}_\alpha$. Hence,
\begin{align}\label{eq1-rp-set}
|L^{(1)}_\alpha|=r+1.\end{align} For each fixed
$\lambda_1'\in\mathbb Z_{r}=\{0,1,\cdots,r-1\}$, by \eqref{T-alf},
we have
\begin{align*}
T(\alpha_{\lambda_1'})=T(\alpha_r)\cup\{1\}.\end{align*} So by
\eqref{L-alf}, we have
\begin{align}\label{eq2-rp-set}
\nonumber&~\mathcal
L(\alpha_{\lambda_1'})\\&=\{(\mu_1,\cdots,\mu_m)\in\mathbb
Z_{r}^m; \mu_1=\lambda_1'\text{~and~}\mu_j=\lambda_j, \forall j\in
T(\alpha_r)\}.\end{align} Moreover, by \eqref{L-alf}, we have
\begin{align}\label{eq3-rp-set}
\mathcal L(\alpha_r)=\{(\mu_1,\cdots,\mu_m)\in\mathbb Z_{r}^m;
\mu_j=\lambda_j, \forall j\in T(\alpha_r)\}.\end{align} Combining
\eqref{eq2-rp-set} and \eqref{eq3-rp-set}, we have
\begin{align}\label{eq4-rp-set}
\mathcal L(\alpha_r)=\bigcup_{\lambda_1'=0}^{r-1}\mathcal
L(\alpha_{\lambda_1'}).\end{align} By construction of $H$ and
$\mathcal C$, for all codeword $(x_1,\cdots,x_n)$ of $\mathcal C$,
we have
\begin{align}\label{eq5-rp-set}
x_{\alpha_{r}}=\sum_{\beta\in\mathcal
L(\alpha_{r})}x_{\beta}\end{align} and for each
$\lambda_1'\in\mathbb Z_{r}=\{0,1,\cdots,r-1\}$, we have
\begin{align}\label{eq6-rp-set}
x_{\alpha_{\lambda_1'}}=\sum_{\beta\in\mathcal
L(\alpha_{\lambda_1'})}x_{\beta}.\end{align} By
\eqref{eq2-rp-set}, $\mathcal L(\alpha_{0}), \mathcal
L(\alpha_{1}), \cdots, \mathcal L(\alpha_{r-1})$ are mutually
disjoint. So by combining \eqref{eq4-rp-set}, \eqref{eq5-rp-set}
and \eqref{eq6-rp-set}, we have
\begin{align}\label{eq7-rp-set} \nonumber
x_{\alpha_r}&=\sum_{\beta\in\mathcal L(\alpha_r)}x_{\beta}\\
\nonumber &=\sum_{\beta\in\bigcup_{\lambda_1'=0}^{r-1}\mathcal
L(\alpha_{\lambda_1'})}x_{\beta}\\
\nonumber & =\sum_{\lambda_1'=0}^{r-1}\left(\sum_{\beta\in\mathcal
L(\alpha_{\lambda_1'})}x_{\beta}\right)\\
&
=\sum_{\lambda_1'=0}^{r-1}x_{\alpha_{\lambda_1'}}
\end{align}
Note that $L^{(1)}_\alpha=\{\alpha_{0},\alpha_{1},
\cdots,\alpha_{r}\}$ and $\alpha=\alpha_{\lambda_1}\in
L^{(1)}_\alpha$. Then by \eqref{eq7-rp-set}, we have
$$x_{\alpha}=\sum_{\alpha'\in
L^{(1)}_\alpha\backslash\{\alpha'\}}x_{\alpha'}.$$ Hence,
$L^{(1)}_\alpha\backslash\{\alpha\}$ is a repair set of $\alpha$.

For any $i\in[m]$, by the same discussion, we can prove that
$L^{(i)}_\alpha\backslash\{\alpha\}$ is a repair set of $\alpha$.
\end{proof}

We give an example as below to show the arguments in the proof of
Lemma \ref{rp-set}.
\begin{exam}\label{eg-1}
Let $r=2$, $m=6$, $\alpha=(0,1,2,0,2,2)$ and $i=4$. Then we have
$$L^{(i)}(\alpha)=\{\alpha_0, \alpha_1, \alpha_2\},$$ where
$\alpha_0=(0,1,2,0,2,2), \alpha_1=(0,1,2,1,2,2)$ and
$\alpha_2=(0,1,2,2,2,2)\}.$ By \eqref{L-alf}, we have
$$T(\alpha_0)=T(\alpha_1)=\{1,2,4\}\text{~and~}T(\alpha_2)=\{1,2\}$$
Moreover, by \eqref{L-alf}, we have $$\mathcal
L(\alpha_0)=\{(0,1,\lambda_3,0,\lambda_5,\lambda_6); \lambda_3,
\lambda_5, \lambda_6\in\mathbb Z_2\},$$ $$\mathcal
L(\alpha_1)=\{(0,1,\lambda_3,1,\lambda_5,\lambda_6); \lambda_3,
\lambda_5, \lambda_6\in\mathbb Z_2\}$$ and $$\mathcal
L(\alpha_2)=\{(0,1,\lambda_3,\lambda_4,\lambda_5,\lambda_6);
\lambda_3, \lambda_4, \lambda_5, \lambda_6\in\mathbb Z_2\}.$$ So
$\mathcal L(\alpha_0)\cap\mathcal L(\alpha_1)=\emptyset$ and
$\mathcal L(\alpha_0)\cup\mathcal L(\alpha_1)=\mathcal
L(\alpha_2)$.

Let $H$ be constructed by \eqref{def-H} and $\mathcal C$ be the
code with parity check matrix $H$. Then for all
$(x_1,\cdots,x_n)\in\mathcal C$, we have
\begin{align*}x_{\alpha_2}&=\sum_{\beta\in\mathcal L(\alpha_2)}x_{\beta}\\
&=\sum_{\beta\in\mathcal L(\alpha_0)\cup\mathcal
L(\alpha_1)}x_{\beta}\\&=\sum_{\beta\in\mathcal
L(\alpha_0)}x_{\beta}+\sum_{\beta\in\mathcal
L(\alpha_1)}x_{\beta}\\&=x_{\alpha_0}+x_{\alpha_1}.\end{align*}
\end{exam}
So $\{\alpha_0,\alpha_1\}$ is a repair set of $\alpha_2$.
Similarly, $\{\alpha_1,\alpha_2\}$ is a repair set of $\alpha_0$,
and $\{\alpha_0,\alpha_2\}$ is a repair set of $\alpha_1$.

Now, we can prove Theorem \ref{main-th}.
\begin{proof}[Proof of Theorem \ref{main-th}]
Note that the $n$ coordinates of codewords of $\mathcal C$ can be
indexed by $\mathbb Z_{r+1}^{m}$. By Lemma \ref{lem-ELRC}, we need
to prove that for any $E\subseteq \mathbb Z_{r+1}^{m}$ of size
$0<|E|\leq2^m-1$, there exists an $\alpha\in E$ such that $\alpha$
has a repair set $R\subseteq\mathbb Z_{r+1}^{m}\backslash E$.
Further, by Lemma \ref{rp-set}, it is sufficient to prove that
there exists an $i\in[m]$ and an $\alpha\in E$ such that
$L^{(i)}_\alpha\backslash\{\alpha\}\subseteq\mathbb
Z_{r+1}^{m}\backslash E$. We can prove this claim by induction on
$m$.

Clearly, the claim is true for $m=1$. To prove the claim for
$m\geq 2$, by induction, we can assume that the claim is true for
$m-1$. That is, for any subset $E'\subseteq \mathbb Z_{r+1}^{m-1}$
of size $0<|E'|\leq2^{m-1}-1$, there exist an $i\in[m-1]$ and an
$\alpha'=(\lambda_1, \cdots, \lambda_i, \cdots, \lambda_{m-1})\in
E'$ such that
$L^{(i)}_{\alpha'}\backslash\{\alpha'\}\subseteq\mathbb
Z_{r+1}^{m-1}\backslash E'$, i.e., $(\lambda_1, \cdots,
\lambda'_{i}, \cdots, \lambda_{m-1})\notin E'$ for all
$\lambda'_{i}\in\mathbb Z_{r+1}\backslash\{\lambda_{i}\}$. Then we
can prove the claim for $m$ as follows.

For each fixed $\lambda\in\mathbb Z_{r+1}$, denote
$$E_{\lambda}=\{(\mu_1,\cdots,
\mu_{m-1}, \mu_m)\in E; \mu_m=\lambda\}.$$ Clearly, the subsets
$E_0,E_1,\cdots,E_r$ are mutually disjoint and
$\bigcup_{j=0}^{r}E_j=E$. We have the following two cases:

Case 1: $0<|E_{\lambda}|\leq2^{m-1}-1$ for some $\lambda\in\mathbb
Z_{r+1}$. Let $$E'=(\mu_1,\cdots, \mu_{m-1})\in\mathbb
Z_{r+1}^{m-1}; (\mu_1,\cdots, \mu_{m-1}, \lambda)\in
E_\lambda\}.$$ Then $0<|E'|=|E_{\lambda}|\leq2^{m-1}-1$. By
induction assumption, there exist an $i\in[m-1]$ and an
$\alpha'=(\lambda_1,\cdots,\lambda_{i},\cdots,\lambda_{m-1})\in
E'$ such that
$(\lambda_1,\cdots,\lambda_{i}',\cdots,\lambda_{m-1})\notin E'$
for all $\lambda_{i}'\in\mathbb Z_{r+1}\backslash\{\lambda_{i}\}$.
So $(\lambda_1,\cdots,\lambda_{i}',\cdots,\lambda_{m-1},
\lambda)\notin E_\lambda$. Note that $E_0,E_1,\cdots,E_r$ are
mutually disjoint and $\bigcup_{j=0}^{r}E_j=E$. Then
$(\lambda_1,\cdots,\lambda_{i}',\cdots,\lambda_{m-1},
\lambda)\notin E$ for all $\lambda_{i}'\in\mathbb
Z_{r+1}\backslash\{\lambda_{i}\}$. Let
$\alpha=(\lambda_1,\cdots,\lambda_i,\cdots\lambda_{m-1},
\lambda)$. Then $\alpha\in E$ and we have
$L^{(i)}_\alpha\backslash\{\alpha\}\subseteq\mathbb
Z_{r+1}^{m}\backslash E$.

Case 2: $|E_\lambda|\geq2^{m-1}$ or $|E_\lambda|=0$ for all
$\lambda\in\mathbb Z_{r+1}$. Since $0<|E|\leq2^m-1$, there exist a
$\lambda_m\in\mathbb Z_{r+1}$ such that $|E_{\lambda_m}|\geq
2^{m-1}$ and $|E_{\lambda}|=0$ for all $\lambda\in\mathbb
Z_{r+1}\backslash\{\lambda_m\}$. Hence, $E\subseteq E_{\lambda_m}$
and $E=E_{\lambda_m}$. Now, let $i=m$ and pick an
$\alpha=(\lambda_1,\cdots, \lambda_{m-1}, \lambda_m)\in
E_{\lambda_m}$. Then $(\lambda_1,\cdots, \lambda_{m-1},
\lambda_m')\notin E_{\lambda_m}=E$ for all $\lambda_m'\in\mathbb
Z_{r+1}\backslash\{\lambda_m\}$. So we have
$L^{(m)}_\alpha\backslash\{\alpha\}\subseteq\mathbb
Z_{r+1}^{m}\backslash E$.

In both cases, there exists an $i\in[m]$ and an $\alpha\in E$ such
that $L^{(i)}_\alpha\backslash\{\alpha\}\subseteq\mathbb
Z_{r+1}^{m}\backslash E$.

Thus, by induction, we proved that for any $E\subseteq \mathbb
Z_{r+1}^{m}$ of size $0<|E|\leq2^m-1$, there exists an $i\in[m]$
and an $\alpha\in E$ such that
$L^{(i)}_\alpha\backslash\{\alpha\}\subseteq\mathbb
Z_{r+1}^{m}\backslash E$. By Lemma \ref{rp-set},
$R=L^{(i)}_\alpha\backslash\{\alpha\}$ is a repair set of
$\alpha$. Hence, by Lemma \ref{lem-ELRC}, $\mathcal C$ is an
$(n,k,r,t=2^m-1)$-ELRC, which completes the proof.
\end{proof}

\section{Conclusions}
The class of $(n,k,r,t)$-exact locally repairable codes (ELRC),
which permit local repair for up to $t$ erasures by the sequential
approach, is the most general setting of LRCs with exact repair.
Several subclasses of LRCs that are reported in the literature,
such as codes with locality $r$ and availability $t$, permit local
repair for up to $t$ erasures by parallel approach and are
contained in the class $(n,k,r,t)$-ELRC.

The direct product of $m$ copies of the $[r+1,r]$
single-parity-check code is a family of codes that has locality
$r$ and availability $m$. In this paper, we prove that such codes
are in fact an $(n,k,r,t)$-ELRC with $t=\sum_{i=1}^mi$. We believe
that such codes are optimal in term of code rate.

There still remains much work to be done for $(n,k,r,t)$-ELRC,
such as the code rate bound for $t\geq 4$ and the minimum distance
bound for $t\geq 3$. Also, constructing $(n,k,r,t)$-ELRC with
sufficiently large code rate (or minimum distance) is an
interesting problem.

\end{document}